\documentclass[a4,11pt]{reportN}

\usepackage[english]{babel}    

\usepackage[latin1]{inputenc}   

\usepackage{graphics}
\usepackage{subfigure}
\usepackage{graphicx}
\usepackage{epsfig}
\usepackage[centertags]{amsmath}
\usepackage{graphicx,indentfirst,amsmath,amsfonts,amssymb,amsthm,newlfont}
\usepackage{longtable}
\usepackage[usenames,dvipsnames,svgnames,table]{xcolor}

\usepackage{verbatim}
\usepackage{multicol}

\usepackage{psfrag} 

\usepackage[noadjust]{cite}

\newtheorem{theo}{Theorem}

\begin{document}

\title{The Lower Bound Error for polynomial NARMAX using an Arbitrary Number of Natural Interval Extensions}

\author
    { \Large{Priscila F. S. Guedes}\thanks{pri12\_guedes@hotmail.com} \\
   \Large{M\'arcia L. C. Peixoto}\thanks{marciapeixoto93@hotmail.com} \\ 
   \Large{Al\'ipio M. Barbosa}\thanks {alipiomonteiro@yahoo.com.br} \\ 
 \Large{Samir A. M. Martins}\thanks{martins@ufsj.edu.br} \\  
 \Large{Erivelton G. Nepomuceno}\thanks{nepomuceno@ufsj.edu.br} \\
   {\small \textsuperscript{1,2,4,5}Control and Modelling Group (GCOM), Department of Eletrical Engineering,\\ Federal University of S\~ao Jo\~ao del-Rei, MG, Brazil}\\
   {\small \textsuperscript{3}Centro Universit\'ario Newton Paiva,  Belo Horizonte, MG, Brazil}\\}

\criartitulo


\begin{abstract}
{\bf Abstract}. The polynomial NARMAX (Nonlinear AutoRegressive Moving Average model with eXogenous input) is a  model that represents the dynamics of physical systems. This polynomial contains information from the past of the inputs and outputs of the process, that is, it is a recursive model. In digital computers this generates the propagation of the rounding error. Our procedure is based on the estimation of the maximum value of the lower bound error considering an arbitrary number of pseudo-orbits produced from  different natural interval extensions, and a posterior Lyapunov exponent calculation. We applied successfully our technique for two identified models of the systems: sine map and Duffing-Ueda oscillator. 

\end{abstract}

\section{Introduction}

System identification is one of the most consolidated and relevant fields of study in science. One of the aims of this science is to obtain mathematical models analogous to the phenomena observed in nature. By analogous systems is meant a system capable of reproducing the characteristics observed in nature. With the identification of systems it is possible to model and investigate systems in an attempt to find some pattern in the observations \cite{Billings2013,aguirre2015}. To identify a system, is necessary to assume a model capable of representing the linear and nonlinear characteristics of the system. Models are mathematical equations that try to describe an approximation of the real system. In the literature there are several ways to identify the same system \cite{NTAA2003}. Different mathematical and computational representations are used, it can be mentioned the neural networks, fuzzy logic, NARMAX (Nonlinear AutoRegressive Moving Average model with eXogenous input) models, among others. The representation of nonlinear systems can be obtained by means of the polynomial NARMAX  \cite{CB1989}. The nonlinear NARMAX polynomials are linear in the parameters, which allows the use of parameter estimation algorithms for linear models \cite{KBLM1988}.  This mathematical representation can be seen as a well-organized recursive function in which the parameters are cautiously chosen \cite{KBLM1988}.

In general, little attention has been given to the propagation of error in the area of system identification, especially in situations that present the polynomial NARMAX. One of the first works related to this subject was \cite{NM2016}. The authors have presented a theorem to estimate the lower bound error for the polynomial NARMAX. In that work, two pseudo-orbits from two different interval extensionswere used to estimate the lower bound error. And in \cite{nm2017}, it was found that the basin of attraction and the invariant distribution were not preserved, the authors shows that from two natural interval extensions may result in differents trajectories, using the lower bound error. 
Thus, this work aims to estimate the lower bound error for $n$ pseudo-orbits derived from $n$ different interval extensions. The proposed method is applied in two identified models of the systems: sine map and Duffing-Ueda oscillator. Afterwards, the  Lyapunov exponent, which is a parameter that characterizes the attractor dynamics. It measures the rate of divergence of neighboring orbits within the attractor, quantifying the dependence or sensitivity of the system relative to the initial conditions \cite{wolf1985}, is calculated and compares this result with the present in the literature for the method presented here and for the one proposed by \cite{MN2016}.

The rest of the paper is organized as follows. In Section \ref{sec:cp} we recall some preliminary concepts of lower bound error and representation the nonlinear systems. Then, in Section 3, we present the developed method. Section 4 is devoted to present the results, then the final remarks are given in Section 5.

 \section{Preliminary concepts}
 \label{sec:cp}
\subsection{The polynomial NARMAX}

The NARMAX model is a representation for nonlinear systems.
This model can be represented as  \cite{CB1989} 
\begin{equation}
 y(k) = F^l \left [y_{k\text{-}1},\cdots ,y_{k\text{-}n_y}, u_{k\text{-}1},  \cdots , u_{k\text{-}n_u},e_{k\text{-}1}, \cdots , e_{k\text{-}n_e} \right ]+ e_{k},  
\end{equation}
where $y_{k}$, $u_{k}$ e $e_{k}$ are, respectively, the output, the input and the noise terms at the discrete time $n \in \mathbb{N}$. The parameters $n_y$, $n_u$ e $n_e$ are their maximum lag. And $F^{\ell}$ is a nonlinear function of degree $\ell$.

\subsection{Recursive functions}

In recursive functions is possible to calculate the state $x_{n+1}$, at a give time, from an earlier state $x_n$
\begin{equation}
   x_{n+1}=f(x_n), 
\end{equation}
where $f$ is a recursive function and $x_n$ is a function state at the discrete time n. Given an initial condition $x_0$, successive applications of the function $f$ it is possible to know the sequence $\{x_n\} $. This sequence can be represented by $\{x_n\}=[x_0,x_1,\cdots,x_n]$ and is defined as an orbit.

Using the computer to calculated the recursive functions, numeric errors are propagated during successive calculations, then the true orbit is not calculated but a representation of the same that is called pseudo-orbit.
\begin{small}
\begin{equation}
\{\hat{x}_{i,n}\}=[\hat{x}_{i,0},\hat{x}_{i,1},\cdots,\hat{x}_{i,n}] \quad\mbox {such that} \quad {|x_n-\hat{x}_{i,n}|\le\delta_{i,n}} \label{eq:orb.1}
\end{equation}
\end{small}
where $\delta_{i,n} \in \mathbb{R}$ is an error and $\delta_{i,n}\ge0$.
So, we may define an interval associated with each value of a pseudo-orbit $I_{i,n}=[\hat{x}_{i,n}-\delta_{i,n}\,,\,\hat{x}_{i,n}+\delta_{i,n}].$ Thus
\begin{equation}
    x_n \in I_{i,n} \quad\mbox{for all i} \in \mathbb{N}.
\end{equation}


\subsection{Natural interval extension}
The natural interval extension is achieved by changing the sequence of arithmetic operation \cite{moore2009}, that is, the extensions are mathematically equivalents.

Furthermore, two extension which algebraically is the same function may not be equivalent in interval arithmetic.

\subsection{Lower bound error}

The lower bound error was proposed by \cite{NM2016}. It is a practical tool capable of increasing the reliability of the computational simulation of dynamic systems.
\begin{theo}
Let two pseudo-orbits $\{\hat{x}_{a,n}\}$ and $\{\hat{x}_{b,n}\}$ derived from two interval extensions. Let $\delta_{\alpha,n}=| \hat{x}_{a,n} - \hat{x}_{b,n}|/2$ be the lower bound error of a map $f(x)$, then $\delta_{a,n}\ge\delta_{\alpha,n}$ or $\delta_{b,n}\ge\delta_{\alpha,n}$.
\end{theo}
The proof of this theorem can be found in \cite{NM2016}.

\section{Methods}

This section is an extension of the work of \cite{NM2016}. The authors developed the Lower bound error theorem for two pseudo-orbits from two different interval extensions. But, for a same map may exist more than two natural interval extensions, so the objective is to investigate the behavior of the natural interval extensions in the computer, exploring the effect of interval dependence, due to the repeated presence of a same interval variable in an algebraic expression, then check on the lower bound error of the pseudo-orbits derived from the different natural interval extensions and  calculate the Lyapunov exponent of the lower bound error from $n$ pseudo-orbits and compare this value. It was clear that the function has more than two extensions, that is, can be rewritten in different ways.

The proposed method can be summarized in the following steps:
\begin{enumerate}
    \item Choose the natural interval extensions;
    \item Calculate the sequence of points of each system from the chosen extensions;
    \item Determine the lower bound error from the combination of two functions;
    \item Determine the maximum lower bound error;
    \item Calculate the Lyapunov exponent and compare the result with the present in literature.
\end{enumerate}
\subsection{Generalization of the lower bound error}
The generalization of the lower bound error is presented in the following theorem.
\begin{theo}
Let an arbitrary number $k \in \mathbb{Z^+}$ of pseudo-orbits derived from interval extensions.  $\zeta_{n}=\cfrac{\max |(\hat{x}_{i,n}-\hat{x}_{j,n})|}{2}$ is the lower bound error, subject to $i \neq j$, $i,j \in \mathbb{N}$, $i \leq k$ and $j \leq k$.
\end{theo}
\begin{proof}
The proof is conducted by \textit{reduction ad absurdum}. Conversely, let the distance between two pseudo-orbit given by $\gamma_n=|(\hat{x}_{i,n}-\hat{x}_{j,n})|$ and let us assume that it is possible to have a lower bound error described by  $$\beta_{n}<\cfrac{\max |(\hat{x}_{i,n}-\hat{x}_{j,n})|}{2}.$$ 
\begin{equation*}
\mbox {Then}, \quad I_{i,n}=[\hat{x}_{i,n}-\beta_{i,n}\,,\,\hat{x}_{i,n}+\beta_{i,n}] \quad\mbox {and} \quad I_{j,n}=[\hat{x}_{j,n}-\beta_{j,n}\,,\,\hat{x}_{j,n}+\beta_{j,n}],
\end{equation*}
for all $i$ and $j$. If it is true, considering the two pseudo-orbits, let us say, $a$ and $b$, for which we have maximum distance between them, it implies that  $I_{a,n} \cap I_{b,n}=\emptyset$ which is a contradiction. And that completes the proof.
\end{proof}

It is clear that for the case of two pseudo-orbits, this theorem is equivalent to that one proposed by  \cite{NM2016}.

\section{Results}
In this section, we present the lower bound error applied for two case studies, which exhibit nonlinear dynamics. We select some natural interval extensions that are equivalent. The two chosen models are for the systems sine map and Duffing-Ueda and all routines are performed in Matlab R2016a in a double precision. We used a computer with a processor Dual core @ 2.7GHz and a Windows 8.1 Professional operating system.

\subsection{Sine map}

A unidimensional sine map is defined as
\begin{equation}
    x_{n+1}=\alpha \sin(x_n),
    \label{eq:sine}
\end{equation}
where $\alpha=1.2\pi$. A polynomial NARMAX identified for this system is given by \cite{NTAA2003}
\begin{equation}
    y_{n+1}=2.6868y_n-0.2462y^3_n.
    \label{eq:1}
\end{equation}
Let us consider four equivalent interval extensions of the model \ref{eq:1}:
\begin{eqnarray}
  F(X_n)&=&2.6868X_n-0.2462X^3_n, \label{eq:2}\\
  G(X_n)&=&2.6868X_n-(0.2462X_n)X^2_n, \label{eq:3} \\
  H(X_n)&=&2.6868X_n-0.2462X_nX_nX_n, \label{eq:4}\\
  L(X_n)&=&X_n(2.6868-0.2462X_nX_n) \label{eq:5}.
 \end{eqnarray}
Equations (\ref{eq:2})-(\ref{eq:5}) are mathematically equivalent, but they represent a different sequence of arithmetic operations. These extensions were simulated using $X_0=0.1$ as an initial condition.
Figure \ref{fig1} shows the evolution of the maximum lower bound error for the sime map and still shows the Lyapunov exponent associated with these extensions using the method developed in \cite{MN2016}. The literature indicates a Lyapunov exponent equals to 1.15 bits/s \cite{NTAA2003}. It is clear that the result presented in Figure \ref{fig1} is in good agreement.

\subsection{Duffing-Ueda}

Considering a damped, periodically forced nonlinear Duffing-Ueda oscillator \cite{Billings2013}:
\begin{equation}
    \frac{d^2y}{dt^2}+k\frac{dy}{dt}+ \mu y^3=A\cos(t).
\end{equation}
where $\mu$ is the cubic stiffness parameter, $k$ is a linear damping and $A$ is the amplitude of excitation. A polynomial NARMAX for the Duffing-Ueda oscillator was identified by \cite{AB1994}.
\begin{small}
\begin{eqnarray}
y_{n+1}&=& 2.1579y_n-1.3203y_{n-1}+0.16239y_{n-2}+0.0003416u_n+0.001963u_{n-1} \nonumber\\ \label{eq:duffing}
&&-0.0048196y^3_n+0.003523y^2_ny_{n-1}-0.0012162y_ny_{n-1}y_{n-2}+0.0002248y^3_{n-2}
\end{eqnarray}
\end{small}
\noindent where $u=A \cos(kT_s)$, $n \in \mathbb{N}$ and $T_s= \pi /60$. Let us consider four interval extensions of Eq. \eqref{eq:duffing}:

\begin{footnotesize}
\begin{eqnarray}
F(X_n)&=& 2.1579X_n\text{-}1.3203X_{n\text{-}1}+0.16239X_{n\text{-}2}+0.0003416U_n+0.001963U_{n\text{-}1} \nonumber\\ 
&&-0.0048196X^3_n+0.003523X^2_nX_{n\text{-}1}\text{-}0.0012162X_nX_{n\text{-}1}X_{n\text{-}2}+0.0002248X^3_{n\text{-}2}\\ \label{eq:duffing1}
G(X_n)&=& 0.0003416U_n+0.001963U_{n\text{-}1}+2.1579X_n\text{-}1.3203X_{n\text{-}1}+0.16239X_{n\text{-}2} \nonumber\\ 
&&-0.0048196X^3_n+0.003523X^2_nX_{n\text{-}1}-0.0012162X_nX_{n\text{-}1}X_{n\text{-}2}+0.0002248X^3_{n\text{-}2}\\ \label{eq:duffing2}
H(X_n)&=& 0.0003416U_n+0.001963U_{n\text{-}1}+2.1579X_n-1.3203X_{n\text{-}1}+0.16239X_{n\text{-}2}-0.0048196X^3_n \nonumber\\ 
&&+0.003523X^2_nX_{n\text{-}1}-0.0012162X_nX_{n\text{-}1}X_{n\text{-}2}+0.0002248X_{n\text{-}2}X_{n\text{-}2}X_{n\text{-}2}\\ \label{eq:duffing3}
L(X_n)&=& 2.1579X_n-1.3203X_{n\text{-}1}+0.16239X_{n\text{-}2}+0.0003416U_n+0.001963U_{n\text{-}1}-0.0048196X_n \nonumber\\
&&X_nX_n+0.003523X^2_nX_{n\text{-}1}-0.0012162X_nX_{n\text{-}1}X_{n\text{-}2}+0.0002248X^3_{n\text{-}2}\label{eq:duffing4}
\end{eqnarray}
\end{footnotesize}

Figure \ref{fig2} shows the evolutions of the maximum lower bound error for the Duffing\text{-}Ueda oscillator with the Lyapunov exponent associated. This largest Lyapunov was calculated by method developed in \cite{MN2016} with value 0.1202 for the maximum of the pseudo-orbits. Oncemore, the computation are in good agreement with the values found in literature, which for this model was calculated in 0.115 \cite{AB1994}.
\begin{figure}[ht!]
\center
\subfigure[fig1][Sine Map]{\includegraphics[width=6.5cm,height=4.5cm]{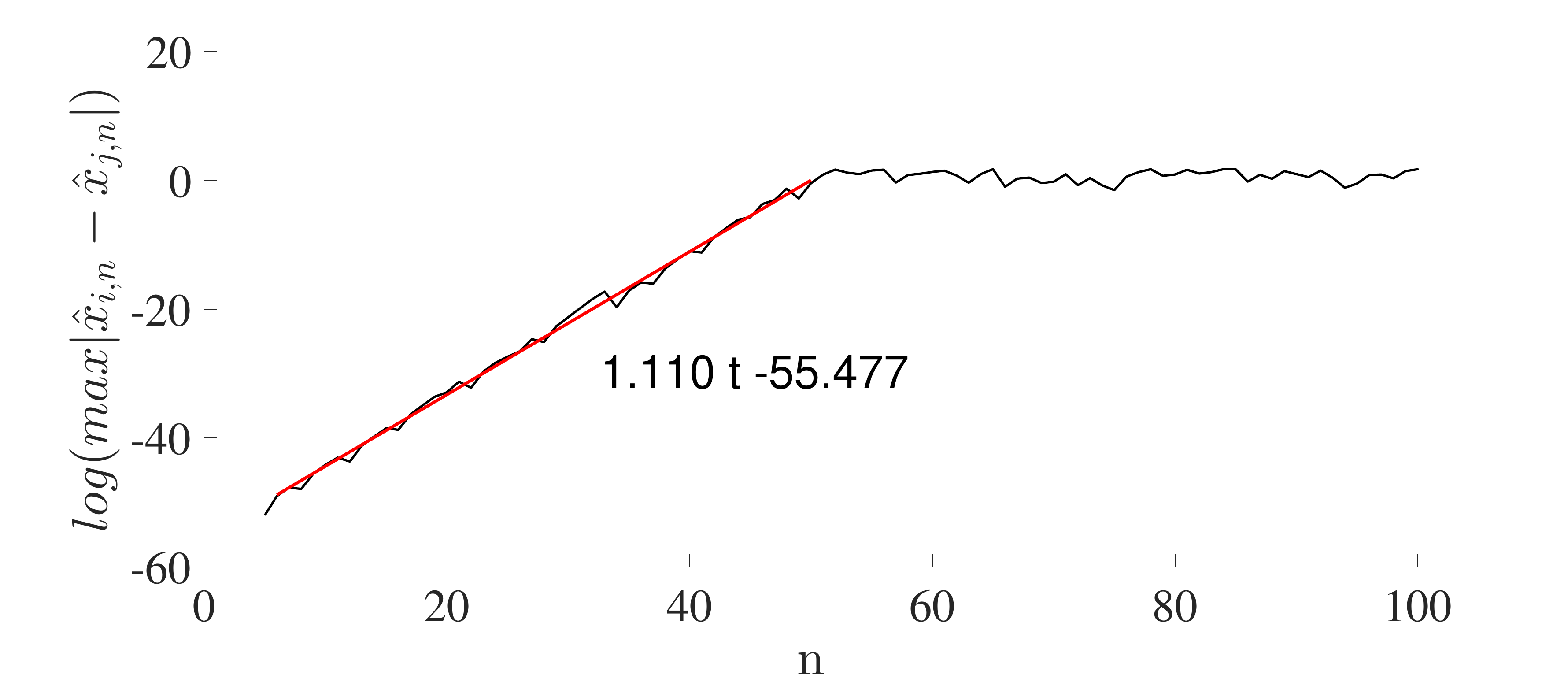}\label{fig1}}
\qquad
\subfigure[fig2][Duffing-Ueda Oscillator]{\includegraphics[width=6.5cm,height=4.5cm]{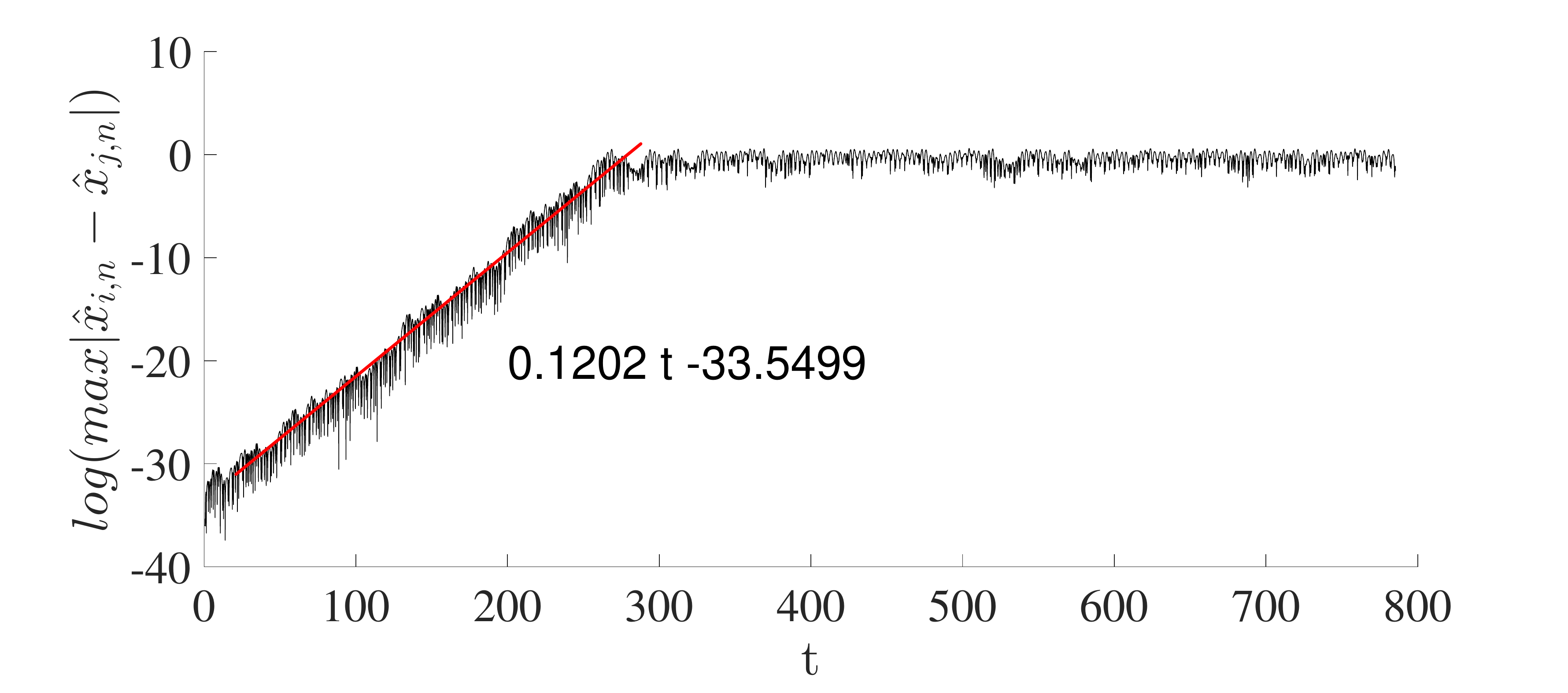}\label{fig2}}
\caption{Evolution of the maximum lower bound error.}
\end{figure}
\section{Conclusion}

The errors present in numerical simulations can lead to an erroneous result that does not correspond to the real situation of the problem. These errors can be of the representation of the model, of the insertion of data, of the numerical algorithm, due to the simplifications, of truncation, of rounding and among others. Thus, some methods have been proposed in order to measure these error, but they are complex from the computational point of view and from the mathematical approximation.

We presented a method to calculate a lower bound error for free-run simulation of the polynomial NARMAX. Our method is based on the comparison of $k$ pseudo-orbits of the same models, but derived from different extension intervals. It makes use of recursive functions, which increases the relevance of this observation.

The methodology was applied in two cases, which are examples of identified systems obtained from literature, by means of the polynomial NARMAX. The sine map and Duffing Ueda oscillator are well known chaotic systems and have been identified using the polynomial NARMAX.

When we compare $k$ pseudo-orbits that are equivalent from the point of view of interval analysis, we proved a theorem that the maximum distance of the pseudo-orbits is greater than the distance of two pseudo-orbits. This maximum value represents a small difference respect to the lower bound error for two pseudo-orbits, but reduces the lower bound error. To prove this statement, the Lyapunov exponent was calculated for the maximum lower bound error, and in both models (sine map and Duffing-Ueda) the result was very close and in good agreement with values calculated in literature.

\section*{Acknowledgments}

We thank CAPES, CNPq/INERGE, FAPEMIG and UFSJ for their support.



\begin{thebibliography}{99}
\bibitem[1]{aguirre2015}
L. A. Aguirre, Introdu{\c{c}}{\~a}o {\`a} identifica{\c{c}}{\~a}o de sistemas - T{\'e}cnicas lineares e n{\~a}o-lineares aplicadas a sistemas reais, Editora UFMG, Edi{\c{c}}{\~a}o 4, (2015).

\bibitem[2]{AB1994}
L. A. Aguirre and S. A. Billings, Validating identified nonlinear models with chaotic dynamics, International Journal of Bifurcation and Chaos, vol. 4, 109-125, (1994).

\bibitem[3]{Billings2013}
S. A. Billings, Nonlinear system identification: NARMAX methods in the time, frequency, and spatio-temporal domains, John Wiley \& Sons, (2013).

\bibitem[4]{CB1989}
S. Chen and S. Billings, Representations of non-linear systems: the NARMAX model, International Journal of Control, vol. 49, 1013-1032, (1989).
\bibitem[5]{KBLM1988}
M. Korenberg, S. A. Billings, Y. P. Liu and P. J. Mcllroy, Orthogonal parameter estimation algorithm for non-linear stochastic systems, International Journal of control, vol. 48, 193-210, (1988), DOI: 10.1080/00207178808906169.

\bibitem[6]{MN2016}
E. M. A. M. Mendes and E. G. Nepomuceno, A very simple method to calculate the (positive) Largest Lyapunov Exponent using interval extensions, International Journal of Bifurcation and Chaos, vol. 26, (2016), DOI: 10.1142/S0218127416502266.

\bibitem[7]{moore2009}
R. E. Moore, R. B. Kearfott and M. J. Cloud, Introduction to interval analysis, Society for Industrial and Applied Mathematics, (2009).

\bibitem[8]{NM2016}
E. G. Nepomuceno and S. A. M. Martins, A lower bound error for free-run simulation of the polynomial NARMAX, Systems Science \& Control Engineering, vol. 4, 50-58, (2016), DOI: 10.1080/21642583.2016.1163296.

\bibitem[9]{nm2017}
E. G. Nepomuceno and E. M. A. M. Mendes, On the analysis of pseudo-orbits of continuous chaotic nonlinear systems simulated using discretization schemes in a digital computer, Chaos, Solitons \& Fractals, vol. 95, 21-32, (2017), DOI: 10.1016/j.chaos.2016.12.002.

\bibitem[10]{NTAA2003}
E. G. Nepomuceno, R. H. C. Takahashi, G. F. V. Amaral and L. A. Aguirre, Nonlinear identification using prior knowledge of fixed points: a multiobjective approach, Journal of Bifurcation and Chaos, vol. 13, 1229-1246, (2003), DOI: 10.1142/S0218127403007187.

\bibitem[11]{wolf1985}
A. Wolf, J. B. Swift, H. L. Swinney and J. A. Vastano, Determining Lyapunov exponents from a time series, Physica D: Nonlinear Phenomena, vol. 16, 285-317, (1985).
\end{thebibliography}
\end{document}